\documentclass[copyright,creativecommons]{eptcs}
\providecommand{\event}{ACAC 2009} % Name of the event you are submitting to
\providecommand{\volume}{??}
\providecommand{\anno}{20??}
\providecommand{\firstpage}{1}
\providecommand{\eid}{??}
\usepackage{breakurl}        % Not needed if you use pdflatex only.

\usepackage{psfrag}
\usepackage[dvips]{graphicx}
\usepackage[rflt]{floatflt}
\usepackage{subfigure}
\usepackage{amsmath}
\usepackage{amsthm}
\usepackage{amssymb}
\usepackage{epsfig}

\newtheorem{theorem}{Theorem}[section]
\newtheorem{proposition}{Proposition}[section]
\newtheorem{definition}{Definition}

\newtheorem{lemma}[theorem]{Lemma}

\theoremstyle{definition}

\newcommand{\dof}{\backslash}

\title{Regular Matroids with Graphic Cocircuits}
\author{
Konstantinos Papalamprou
\institute{Operational Research Group,
Department of Management \\
London School of Economics,
London, UK}
\email{k.papalamprou@lse.ac.uk}
\and
Leonidas Pitsoulis
\institute{Department of Mathematical and Physical Sciences,\\
Aristotle University of Thessaloniki,
Thessaloniki, Greece}
\email{pitsouli@gen.auth.gr}
}
\def\titlerunning{Regular Matroids with Graphic Cocircuits}
\def\authorrunning{K. Papalamprou \& L. Pitsoulis}
\begin{document}
\maketitle

\begin{abstract}
We introduce the notion of graphic cocircuits and show that a large class of regular matroids with graphic cocircuits belongs to the class of signed-graphic matroids. Moreover, we provide an algorithm which determines whether a cographic matroid with graphic cocircuits is signed-graphic or not.
\end{abstract}

%%%%%%%%%%%%%%%%%%%%%%%%%%%%%%%%%%%%%%%%%%%%%%%%%%%%%%%%%%%%%%%%%
\section{Introduction}
%%%%%%%%%%%%%%%%%%%%%%%%%%%%%%%%%%%%%%%%%%%%%%%%%%%%%%%%%%%%%%%%%
In this paper we examine the effect of removing cocircuits from regular matroids and we focus on the case in which such a removal always results in a graphic matroid. The first main result, given in section~\ref{sec_excl_mi}, is that a regular matroid with graphic cocircuits is signed-graphic if and only if it does not contain two specific minors. This provides a useful connection between graphic, regular and signed-graphic matroids which may be further utilised for devising combinatorial recognition algorithms for certain classes of matroids. At this point we should note that decomposition theories and recognition algorithms for matroids have provided some of the most important results of matroid theory and combinatorial optimization (see e.g. decomposition of graphic matroids~\cite{Tutte:1960} and recognition of network matrices~\cite{BixCunn:1980} and decomposition of regular matroids~\cite{Seymour:1980} and recognition of totally unimodular matrices~\cite{Seymour:95}). Finally, in section~\ref{sec_recal} we provide a simple recognition algorithm which determines whether a cographic matroid with graphic cocircuits is signed-graphic or not.  

%%%%%%%%%%%%%%%%%%%%%%%%%%%%%%%%%%%%%%%%%%%%%%%%%%%%%%%%%%%%%%%%%
\section{Preliminaries} \label{sec_preliminaries}
%%%%%%%%%%%%%%%%%%%%%%%%%%%%%%%%%%%%%%%%%%%%%%%%%%%%%%%%%%%%%%%%%

In this section we will mention all the necessary definitions and preliminary results regarding graphs, signed graphs and their
corresponding matroids. The definitions for graphs presented in this section are taken from~\cite{Diestel:05,Tutte:01},
while for signed graphs from~\cite{Zaslavsky:1982,Zaslavsky:1989}. The main reference for matroid theory is the book of
Oxley~\cite{Oxley:06} while the main reference for signed-graphic matroids is~\cite{Zaslavsky:1991a}.

\subsection{Graphs} \label{subsec_graphs}%%%%%%%%%%%%%%%%%%%%%%%%%%%%%%%%%%%%%%%%%%%%%%%%

A graph $G:=(V,E)$ is defined as a finite set of vertices $V$, and a set of
edges $E\subseteq V\cup V^{2}$ where identical elements are allowed. 
Therefore we will have four types of edges: $e=\{u, v \}$ is called
a \emph{link}, $e=\{v, v \}$ a \emph{loop}, $e=\{v\}$ a \emph{half edge}, 
while $e=\emptyset$ is a \emph{loose edge}. Whenever applicable, the vertices
that define an edge are called its \emph{end-vertices}. We say that a vertex 
$v$ of a graph $G$ is \emph{incident with}  an edge $e$ of $G$ and that $e$ is \emph{incident with} $v$ if $v\in e$. We also say that two vertices $u$ and $v$ of $G$ are \emph{adjacent}  or that $u$ is \emph{adjacent to $v$} if $\{u,v\}$ is an edge of $G$.
Observe that the above is the ordinary definition of a graph, except that we
also allow half edges and loose edges. We will denote the set of vertices and the set of edges
of a graph $G$ by $V(G)$ and $E(G)$, respectively.

In what follows we will assume that we have a graph $G$. 
The following operations are defined. 
We say that $G'$ is a \emph{subgraph} of $G$, denoted by
$G'\subseteq{G}$, if $V(G')\subseteq{V(G)}$ and
$E(G')\subseteq{E(G)}$. 
For some $X\subseteq V(G)$ the subgraph \emph{induced} by $X$ is defined 
as $G[X] := (X, E')$, where $E'\subseteq E(G)$ is a maximal set of edges with all
end-vertices in $X$. If $X\subseteq E(G)$ then $G[X] := (V',X)$, where $V'\subseteq V(G)$
is the set of end-vertices of edges in $X$.
The \emph{deletion of an edge} $e$ from $G$ is the subgraph defined as 
$G\dof e := (V(G), E-e)$. \emph{Identifying}  two vertices $u$ and $v$ is the operation 
where we replace $u$ and $v$ with a new vertex $v'$ in both $V(G)$ and $E(G)$. 
The \emph{contraction of a link} $e=\{u,v\}$ is the subgraph, denoted by
$G/e$, which results from $G$ by identifying $u,v$ in $G\dof e$.
The \emph{contraction of a half edge} $e=\{v\}$ or a \emph{loop} $e=\{v,v\}$ is the subgraph, denoted by
$G/e$, which results from the removal of $\{v\}$ and all half edges and loops incident with it, while all links
incident with $v$ become half edges at their other end-vertex. Contraction of a loose edge is the same as deletion. 
The \emph{deletion of a vertex} $v$ of $G$ is defined as the deletion of all edges incident with $v$ and 
the deletion of $v$ from $V(G)$.  
A graph $G'$ is called a \emph{minor} of $G$ if it is obtained from a sequence of deletions and contractions 
of edges and deletions of vertices of $G$.

Two graphs $G$ and $H$ are called \emph{isomorphic}, and we write $G\cong H$, if there exists
a bijection $p:V(G)\rightarrow V(H)$ such that $\{u,v\} \in E(G)$ if and only if $\{p(u),p(v)\}\in E(H)$. A \emph{walk} in $G$ is a sequence  $(v_1, e_1, v_2, e_2,\ldots, e_{t-1}, v_t)$ where 
$e_i$ is incident with both $v_i$ and $v_{i+1}$.  If $v_1=v_t$, then we say that the walk is \emph{closed}. 
If a walk has distinct inner vertices, then it is called a \emph{path}.  The subgraph  of $G$ induced by the edges of a closed path is called a \emph{cycle}. The edge set of a cycle of $G$ is called a \emph{circle} of $G$. A graph is called a \emph{wheel graph}, if it consists of a cycle along with a vertex which is adjacent to every vertex of the cycle. 

A graph is \emph{connected} if there is a walk between any pair of its vertices. 
There are several notions of  higher connectivity in graphs that have appeared in the literature.  
Here we will define  \emph{Tutte $k$-connectivity} which we shall call simply \emph{$k$-connectivity}. 
For $k\geq 1$, a \emph{$k$-separation} of a connected graph $G$ is a partition $(A,B)$ of the edges 
such that $\min\{|A|,|B|\}\geq{k}$ and $|V({G[A]}) \cap V({G[B]})|=k$. 
For $k\geq 2$, we say that $G$ is \emph{$k$-connected}  
if $G$ does not have an $l$-separation for $l=1,\ldots,k-1$. 
Note that our notion of $k$-connectivity of a graph is taken from Tutte's graph theory book~\cite{Tutte:01} which is different from the notion of $k$-connectivity we find in other graph theory books, e.g.~\cite{BonMur:07,Diestel:05}. We use $k$-connectivity as defined above in this paper, due to the fact that the connectivity of a graph and its corresponding graphic matroid coincide under this definition.

\subsection{Signed graphs}\label{subsec_signed_graphs}%%%%%%%%%%%%%%%%%%%%%%%%%%%%%%%%%%%%%%%%%%%%%%%%%%

A \emph{signed graph} is defined as $\Sigma := (G,\sigma)$ where $G$ is a graph called the \emph{underlying graph} and
$\sigma$ is a sign function $\sigma:E(G)\rightarrow \{\pm 1\}$, where 
$\sigma(e) =-1$ if $e$ is a half edge and  $\sigma(e) =+1$ if $e$ is a loose edge. 
Therefore a signed graph is a graph where the  edges are labelled as  positive or  negative, while all the half edges
are negative and all the loose edges are positive. 
We denote by $V(\Sigma)$ and $E(\Sigma)$ the vertex set and edge set of a signed graph $\Sigma$, respectively. 

 The \emph{sign of a cycle} is the product of the signs of its edges, so we have a 
\emph{positive cycle} if the number of negative edges in the cycle is  even, otherwise the cycle is a \emph{negative cycle}.
Both negative loops and  half-edges are negative cycles. A signed graph is 
called \emph{balanced}  if it contains no negative  cycles. Finally, although signed graphs have been studied extensively, it is out of the scope of this work to provide more notions, definitions and results regarding signed graphs. However, the interested reader is referred to ~\cite{Zaslavsky:1982,Zaslavsky:1991a,Zaslavsky:1999}

\subsection{Matroids}%%%%%%%%%%%%%%%%%%%%%%%%%%%%%%%%%%%%%%%
\begin{definition} \label{def_ntefia}
A matroid $M$ is an ordered pair $(E,\mathcal{I})$  of a finite set $E$  and a collection $\mathcal{I}$ of subsets of $E$ satisfying the following three conditions:
\begin{itemize}
\item [(I1)] $\emptyset \in{\mathcal{I}}$
\item[(I2)] If $X\in{\mathcal{I}}$ and $Y\subseteq{X}$ then $Y\in{\mathcal{I}}$ 
\item[(I3)] If $U$ and $V$ are members of $\mathcal{I}$ with $|U|<|V|$ then there exists $x\in{V-U}$ such that $U\cup{x}\in{\mathcal{I}}$. 
\end{itemize}
\end{definition}

Given a matroid $M=(E,\mathcal{I})$, the set $E$ is called the \emph{ground set} of $M$ and the members of $\mathcal{I}$  are the \emph{independent sets} of $M$; furthermore, any subset of $E$ not in $\mathcal{I}$ is called a $\emph{dependent set}$ of $M$. A  minimal dependent set is called a \emph{circuit} of $M$. The \emph{rank function} $r_M:2^{E}\rightarrow{\mathbb{Z_+}}$ of a matroid $M$ is a function defined by: $r_M(A)=max(|X|:X\subseteq{A},X\in{\mathcal{I}})$, where $A\subseteq{E}$ and $|A|$ is the cardinality of $A$.
The axiomatic Definition~\ref{def_ntefia} for a matroid on a given ground set uses its independent sets. However, there are several equivalent ways to define a matroid which can be found in \cite{Oxley:06}. For example, a matroid $M$ on a given ground set $E$ can be defined through its rank function or through its set of circuits. We provide here the following axiomatisation of a matroid by its circuits~\cite{Oxley:06}:
\begin{proposition}
A collection $\mathcal{C}$ of subsets of $E$ is the collection of circuits of a matroid on $E$ if and only if $\mathcal{C}$ satisfies the following conditions:
 \begin{itemize}
\item [(I1)] $\emptyset \notin{\mathcal{C}}$
\item[(I2)] If $C_1$ and $C_2$ are members of $\mathcal{C}$ and $C_1\subseteq{C_2}$, then $C_1=C_2$. 
\item[(I3)] If $C_1$  and $C_2$ are distinct members of $\mathcal{C}$ and $e\in{C_1\cap{C_2}}$, then there is a $C_3\in{\mathcal{C}}$ such that $C_3\subseteq{(C_1\cup C_2)-e}$.
\end{itemize}
\end{proposition}

Two matroids $M_1$ and $M_2$ are called \emph{isomorphic} if there is a bijection $\psi$ from $E(M_1)$ to $E(M_2)$ such that $X\in{\mathcal{I}(M_1)}$ if and only if $\psi{(X)}\in{\mathcal{I}(M_2)}$. We denote that $M_1$ and $M_2$ are isomorphic by $M_1\cong{M_2}$. 

Let $E$ be a finite set of vectors from a vectorspace over some field $F$ and let $\mathcal{I}$ be the collection of all subsets of linearly independent elements of $E$; then it can be proved that $M=(E,\mathcal{I})$ is a matroid called \emph{vector matroid}. Furthermore, any matroid isomorphic to $M$ is called a \emph{representable matroid} over $F$. Matroids representable over the finite field $GF(2)$ are called \emph{binary}  and matroids representable over the finite field $GF(3)$ are called \emph{ternary}. A matroid representable over every field is a \emph{regular} matroid.
Let $A$ be a matrix whose columns are the vectors of the ground set of a vector matroid $M$. It is evident that there is one-to-one correspondence between the linearly independent columns of $A$ and the independent sets of $M$, so the matroid $M$ can be fully characterized by matrix $A$. Matrix $A$ is called a \emph{representation matrix} of $M$ and we  denote the vector matroid with representation matrix $A$ by $M[A]$.
Let $G$ be a graph without loops, half-edges or loose edges and let $\mathcal{C}$ be the collection of edge sets of cycles of $G$. Then it can be shown that the pair $(E(G), \mathcal{C})$ is a matroid called the \emph{cycle matroid} of $G$  and is denoted by $M(G)$. A matroid $M$ such that $M\cong{M(G)}$ is called \emph{graphic}.

Given a matroid $M=(E,\mathcal{I})$, the ordered pair $(E,\{E-{S}:S \notin{\mathcal{I}}\})$ is a matroid called the \emph{dual matroid} of $M$ and denoted by $M^{*}$. There is always a dual matroid $M^{*}$ associated with a matroid $M$ and it is clear that $(M^{*})^{*}=M$. Usually, the prefix 'co' is used to dualize a term. Therefore, the set $\mathcal{C}(M^{*})$ of circuits  of $M^{*}$ is the set of \emph{cocircuits} of $M$. We usually denote the cocircuit of $M$ by $\mathcal{C}^{*}(M)$.

\emph{Deletion} and \emph{contraction} are two fundamental matroid operations. Formally, given a matroid $M=(E, \mathcal{C})$ on a ground set $E$ defined by its collection of circuits $\mathcal{C}$ the \emph{deletion of}  some $T\subseteq E$ from $M$ is the matroid  denoted by $M\backslash T$, on $E\backslash{T}$ with the following collection of circuits:
\begin{equation} \label{eq_nn1}
\mathcal{C}(M\backslash T) := \{ C\in \mathcal{C}(M) | C\cap{T}=\emptyset \}. 
\end{equation} 
The \emph{contraction of} some $T\subseteq E$ is the matroid denoted by $M/ T$, on $E\backslash{T}$ with the following collection of circuits: 
\begin{equation} \label{eq_nn2}
\mathcal{C}(M/ T) := \mbox{minimal}\{ C\backslash{T} | C\in\mathcal{C}(M)\}.
\end{equation}

Furthermore, deletion and contraction may be viewed as dual operations in the sense that the deletion or contraction of a set $T\subseteq{E(M)}$ from $M$ is translated as the contraction or deletion of $T$ from $M^{*}$, respectively. In a symbolic way this is expressed as follows:
\begin{equation} \label{eq_lambra}
M\backslash{T}=(M^{*}/T)^{*} \textrm{ and } M/T=(M^{*}\backslash{T})^{*}
\end{equation}

Any matroid which can be obtained from $M$ by a series of deletions and contractions is called a \emph{minor} of $M$.  If $M$ has a  minor isomorphic to a matroid $N$ then we will often say that $M$ has an \emph{$N$-minor} or $M$ has $N$ \emph{as a minor}. A matroid $N$ is called an \emph{excluded minor} for a class of matroids $\mathcal{M}$  if $N\notin{\mathcal{M}}$ but every proper minor of $N$ is in $\mathcal{M}$. A well-known excluded minor characterization for graphic matroids goes as follows~\cite{Tutte:1965}, where $K_5$ is the complete graph on five vertices and $K_{3,3}$ is the complete bipartite graph having three vertices at each side of the bipartition:
\begin{theorem} \label{th_tttt}
A regular matroid is graphic if and only if it has no minor isomorphic to $M^{*}(K_{3,3})$ or $M^{*}(K_5)$.
\end{theorem}
Consider a matroid $M$ defined by a rank function $r:E(M)\rightarrow \mathbb{Z}$. For some positive integer $k$, a partition $(X,Y)$ of $E(M)$ is called a \emph{$k$-separation} of $M$ if the following two conditions are satisfied:
\begin{itemize}
\item[(M1)] $\min\{|X|,|Y|\}\geq k$, and
\item[(M2)] $r_{M}(X)+r_{M}(Y)-r(M) \leq k-1$.
\end{itemize}
If a matroid $M$ is $k$-separated for some integer $k$ then the \emph{connectivity of a matroid} of $M$ is the smallest integer $j$ for which $M$ is $j$-separated; otherwise, we take the connectivity of $M$ to be infinite.

\subsection{Signed-Graphic Matroids} \label{subsec_sign_graph_matrds}%%%%%%%%%%%%%%%%%%%%%%%%%%%%%%%%%%%%%%%%%%%%%%%%%%  
The definition  of the \emph{signed-graphic matroid} goes as follows~\cite{Zaslavsky:1982}: 
\begin{theorem} \label{th_sgg2}
Given a signed graph $\Sigma$ let $\mathcal{C} \subseteq 2^{E(\Sigma)}$ be the family of minimal edge sets 
inducing a subgraph in $\Sigma$ which is either:
\begin{itemize}
\item[(a)] a positive cycle, or
\item[(b)]  two negative cycles which have exactly one common vertex, or 
\item[(c)] two vertex-disjoint negative cycles connected by a path which has no common vertex with the cycles apart from 
its end nodes.
\end{itemize}
Then $M(\Sigma)=(E(\Sigma), \mathcal{C})$ is a matroid on $E(\Sigma)$ with circuit family $\mathcal{C}$. 
\end{theorem}
\noindent
The subgraphs of $\Sigma$ induced by the edges corresponding to a circuit of $M(\Sigma)$ are called the 
\emph{circuits} of $\Sigma$. Therefore a circuit of $\Sigma$ can be one of three types (see Figure~\ref{fig_circuits_signed_graphs} for example circuits of types (a), (b) and (c)). 

\begin{figure}[hbtp]
\begin{center}
\mbox{
\subfigure[]
{
\includegraphics*[scale=0.3]{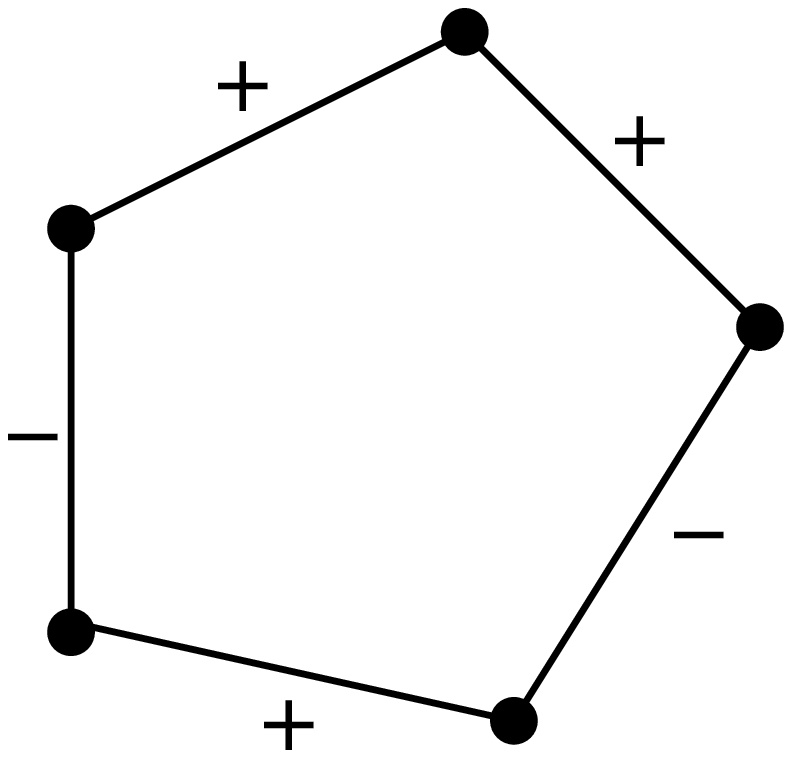}
}\qquad
\subfigure[]
{
\includegraphics*[scale=0.3]{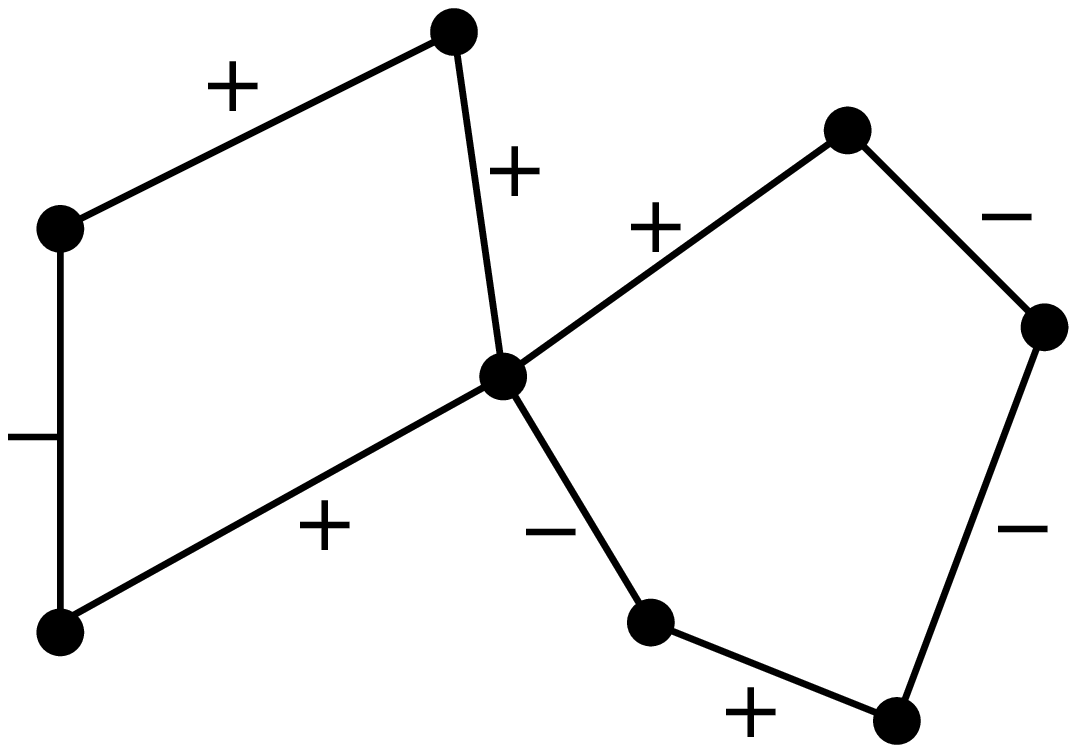}
}\qquad
\subfigure[]
{
\includegraphics*[scale=0.3]{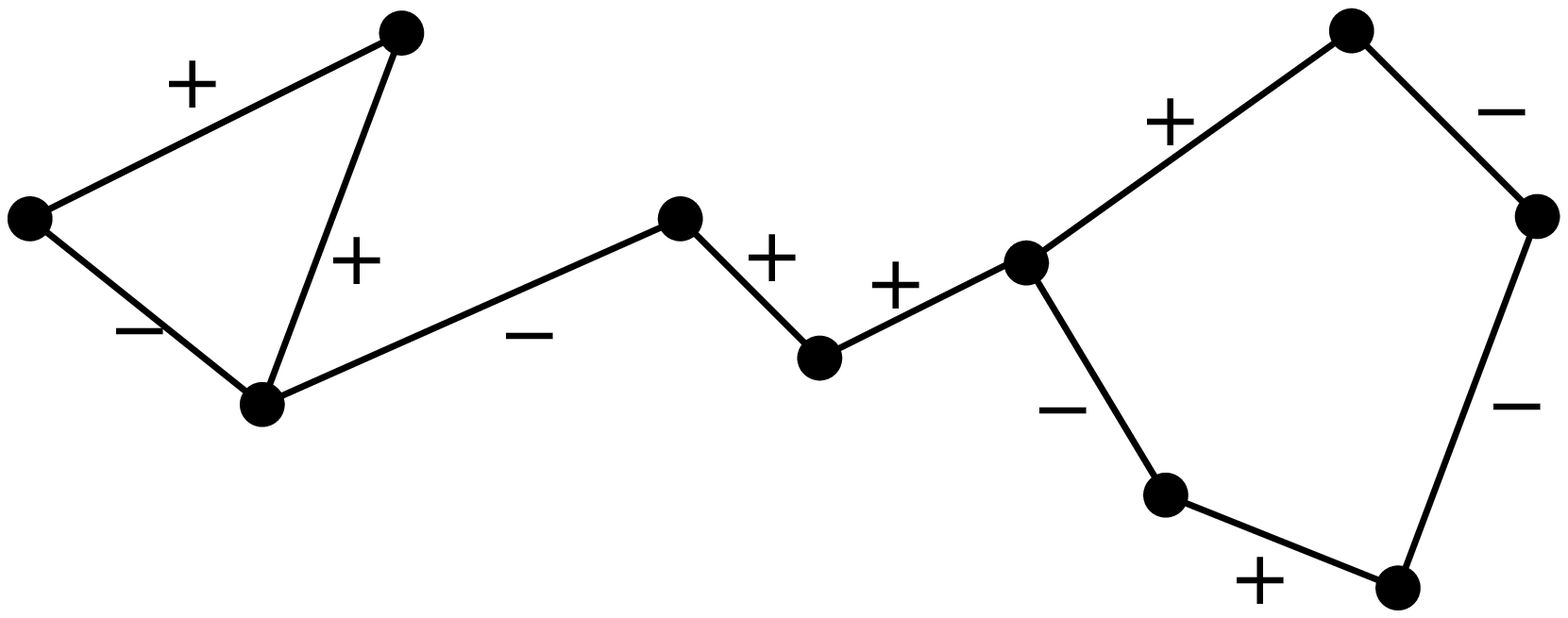}
}} \caption{Circuits in a signed graph $\Sigma$}
\end{center}
\end{figure}\label{fig_circuits_signed_graphs}
  
For each signed graph $\Sigma$ with edge set $E(\Sigma)$, there is an associated signed-graphic
matroid $M(\Sigma)$ on the set of elements $E(\Sigma)$. However for a given signed-graphic matroid $M$ 
there may exist several signed graphs $\Sigma_{i}$ such that $M=M(\Sigma_{i})$ where $i\geq 1$. So 
signed-graphic matroids can be viewed as the abstract entities, while their corresponding signed graphs
their representations in a graphical context.

\section{An excluded minor characterization} \label{sec_excl_mi}%%%%%%%%%%%%%%%%%%%%%%%%%%%%%%%%%%%%%%%%%%%%%%%%%%%%%%%%%%%%
We say that a cocircuit $Y$ of a binary matroid $M$ is \emph{graphic}  if $M\dof Y$ is a graphic matroid; otherwise, we say that  $Y$ is a~\emph{non-graphic} cocircuit.
Two important theorems which associate signed-graphic matroids with cographic matroids and regular matroids in terms of excluded minors have been shown by Slilaty et.al. in~\cite{SliQinZhou:07,Slilaty:2005b}. 
Specifically, of the 35 forbidden minors for projective planar graphs 29 are not 1-separable; these 29 graphs, which we call 
$G_1,G_2,\ldots,G_{29}$, can be found in~\cite{Arch:81, MoTh:01}. Slilaty has shown in~\cite{Slilaty:2005b} that the collection of the cographic matroids of these 29 graphs $\{ M^{*}(G_1),M^{*}(G_2),\ldots,M^{*}(G_{29}) \}$, forms the complete list of the cographic 
excluded minors for signed-graphic matroids. Since cographic matroids is a subclass of regular matroids (see~\cite{Oxley:06}), we 
expect the list of regular excluded minors for signed-graphic matroids to contain the matroids in $\mathcal{M}$ 
and some other matroids. Those other matroids are the $R_{15}$ and $R_{16}$ whose representation matrices over $GF(2)$ 
are the following:

{\scriptsize 
\begin{equation*}
% A_{R_{15}}=
\underbrace{
\left[ \begin{array}{rrrrrrrrrrrrrrr}
1 & 0 & 0 & 0 & 0 & 0 & 0 &         1  & 0  & 1  &  0 & 0  & 0  & 0  & 1  \\
0 & 1 & 0 & 0 & 0 & 0 & 0 &         0  & 0  & 0  &  1 & 1  & 0  & 1  & 0  \\
0 & 0 & 1 & 0 & 0 & 0 & 0 &         1  & 1  & 0  &  0 & 1  & 1  & 0  & 0  \\ 
0 & 0 & 0 & 1 & 0 & 0 & 0 &         1  & 1  & 0  &  0 & 0  & 1  & 1  & 0  \\
0 & 0 & 0 & 0 & 1 & 0 & 0 &         0  & 1  & 1  &  1 & 1  & 0  & 0  & 0  \\ 
0 & 0 & 0 & 0 & 0 & 1 & 0 &         0  & 1  & 1  &  1 & 1  & 1  & 0  & 0  \\
0 & 0 & 0 & 0 & 0 & 0 & 1 &         0  & 1  & 1  &  1 & 1  & 1  & 0  & 1  
\end{array} \right]}_{R_{15}}, 
% A_{R_{16}}=
\underbrace{
\left[ \begin{array}{rrrrrrrrrrrrrrrr}
1 & 0 & 0 & 0 & 0 & 0 & 0 & 0 &     0  & 1  & 1  &  0 & 1  & 0  & 0  & 0  \\
0 & 1 & 0 & 0 & 0 & 0 & 0 & 0 &     0  & 0  & 0  &  0 & 1  & 1  & 1  & 0  \\
0 & 0 & 1 & 0 & 0 & 0 & 0 & 0 &     0  & 1  & 1  &  0 & 1  & 1  & 1  & 0  \\ 
0 & 0 & 0 & 1 & 0 & 0 & 0 & 0 &     0  & 0  & 0  &  1 & 1  & 1  & 0  & 0  \\
0 & 0 & 0 & 0 & 1 & 0 & 0 & 0 &     1  & 1  & 0  &  1 & 1  & 1  & 0  & 0  \\
0 & 0 & 0 & 0 & 0 & 1 & 0 & 0 &     1  & 0  & 1  &  1 & 0  & 0  & 0  & 0  \\
0 & 0 & 0 & 0 & 0 & 0 & 1 & 0 &     1  & 1  & 0  &  0 & 0  & 0  & 0  & 1  \\
0 & 0 & 0 & 0 & 0 & 0 & 0 & 1 &     1  & 0  & 1  &  0 & 0  & 0  & 0  & 1  
\end{array} \right]}_{R_{16}}.
\end{equation*}
}
Specifically, in~\cite{Slilaty:2005b} we find Theorem~\ref{th_ree}  and in~\cite{SliQinZhou:07} we find Theorem~\ref{th_slqz12}.
\begin{theorem} \label{th_ree}
A cographic matroid $M$ is signed-graphic if and only if $M$ has no minor isomorphic to $M^{*}(G_1),\ldots,M^{*}(G_{29})$. 
\end{theorem}
\begin{theorem} \label{th_slqz12}
A regular matroid $M$ is signed-graphic if and only if $M$ has no minor isomorphic to $M^{*}(G_1),\ldots,M^{*}(G_{29}), R_{15}$ 
or $R_{16}$. 
\end{theorem}

The following two lemmas are essential for the proof of the main result of this section which characterizes the regular 
matroids with graphic cocircuits.

\begin{lemma} \label{lem_ll1}
If a matroid  $M$ is isomorphic to $M^{*}(G_{17})$ or  $M^{*}(G_{19})$ then for any cocircuit 
$Y\in{\mathcal{C}^{*}(M)}$, the matroid $M\backslash{Y}$ is graphic. 
\end{lemma}
\begin{proof}
By (\ref{eq_lambra}), we can equivalently show that, for any circuit $Y$ of $M^{*}$, the matroid $M^{*}/Y$ is cographic. The matroid $M^{*}$ is graphic and thus, regular. Therefore, by Theorem~\ref{th_tttt}, we have to show that for any circuit $Y$ of $M^{*}\in\{M(G_{17}),M(G_{19}))\}$ the matroid $M^{*}/Y$ has no minor isomorphic to $M(K_5)$ or $M(K_{3,3})$. Observe that $G_{17}$ is isomorphic to the graph $K_{3,5}$ and $G_{19}$ is isomorphic to $K_{4,4}\backslash{e}$, where $e$ is any edge of $K_{4,4}$. Since $M(G_{19})$ is a graphic matroid  we have that $M(G_{19})\cong{M(K_{4,4}\backslash{e})}=M(K_{4,4})\backslash{e}$. Therefore,  $M(K_{4,4})$ has a minor isomorphic to $M(G_{19})$ and by~\eqref{eq_nn1}, any circuit of $M(G_{19})$ is a circuit of $M(K_{4,4})$. Thus, it suffices to prove that for any circuit $Y_1\in{\mathcal{C}(M(K_{3,5}))}$ and $Y_2\in{\mathcal{C}(M(K_{4,4}))}$ the matroids  $M(K_{3,5})/Y_1=M(K_{3,5}/Y_1)$ and  $M(K_{4,4})/Y_2=M(K_{4,4}/Y_2)$ have no minor isomorphic to $M(K_5)$ or $M(K_{3,3})$ in order to prove the theorem.

Since $K_{3,5}$ and  $K_{4,4}$ are $3$-connected,  we get that $Y_1$ and $Y_2$ correspond to circles of $K_{3,5}$ and  $K_{4,4}$, respectively.  The $3$-connected graphs $K_{3,5}$ and $K_{4,4}$ are also bipartite and therefore, they have  no circle of odd cardinality and moreover, they have no parallel edges. This means that $K_{3,5}/Y_{1}$ and $K_{4,4}/Y_{2}$ have at most five vertices each. Therefore, the matroids  $M(K_{3,5}/{Y_1})$ and  $M(K_{4,4}/Y_2)$ have rank at most $4$ which is less than the rank of $M(K_{3,3})$. Evidently, $M(K_{3,5}/{Y_1})$ and  $M(K_{4,4}/Y_2)$ can not have a minor isomorphic to $M(K_{3,3})$. 

It remains to be shown that $M(K_{3,5}/{Y_1})$ and  $M(K_{4,4}/Y_2)$ have no minor isomorphic to $M(K_5)$. Let us suppose that $Y_1$ and $Y_2$ are circuits of cardinality four. Then, since $K_{3,5}$ and $K_{4,4}$ are $3$-connected we have  that (by Lemma 5.3.2 in~\cite{Oxley:06}) $Y_1$ and $Y_2$ are circles of $K_{3,5}$ and $K_{4,4}$, respectively, with cardinality four. Observe now that for any $Y_1$ and $Y_2$, the graphs $K_{3,5}/Y_1$  and $K_{4,4}/Y_2$ are isomorphic to the graphs $\bar{G}$ and $\hat{G}$ of~Figure~\ref{fig:g2}, respectively. Furthermore, parallel edges of a graph correspond to parallel elements in the associated graphic matroid. Therefore, any simple minor of $M(\bar{G})$ or $M(\hat{G})$ has at most seven or eight elements, respectively. The matroid $M(K_5)$ is simple and has ten elements. Therefore, $M(K_5)$ can not be a minor of $M(\bar{G})\cong{M(K_{3,5}/{Y_1})}$ or $M(\hat{G})\cong{M(K_{4,4}/Y_2)}$.
For the remaining case, that is, if $Y_1$ or $Y_2$ has more than four elements, the proof is quite similar to the one we followed in order to prove that $M(K_{3,5}/{Y_1})$ and  $M(K_{4,4}/Y_2)$ have no minor isomorphic to $M(K_{3,3})$ and for that reason is ommited.
\begin{figure}[h] 
\begin{center}
\centering
\psfrag{G1}{\footnotesize $\bar{G}$}
\psfrag{G2}{\footnotesize $\hat{G}$}
\includegraphics*[scale=0.3]{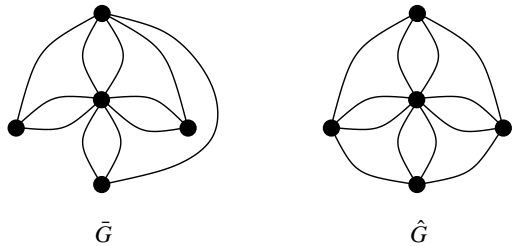}
\end{center}
\caption{The graphs $\bar{G}$ and  $\hat{G}$.}
\label{fig:g2}
\end{figure} 
\end{proof}

\begin{lemma}\label{lem_1}
If $N$ is a minor of a matroid $M$ then for any cocircuit $C_{N}$ of $N$ there
exists a cocircuit $C_{M}$ of $M$ such that $N\backslash C_{N}$ is a minor of 
$M\backslash C_{M}$. 
\end{lemma}
\begin{proof}
We have that  $N=M\backslash X / Y$ for some disjoint $X,Y\subseteq{E(M)}$ and by duality, $N^{*}=M^{*} / X \backslash Y$. Therefore, by the definitions of deletion and contraction given in \eqref{eq_nn1} and \eqref{eq_nn2}, we have
that for any cocircuit $C_{N}\in\mathcal{C}(N^{*})$ of $N$ there exists a cocircuit $C_{M}\in\mathcal{C}(M^{*})$ of $M$ such
that 
\begin{itemize}
\item[(i)] $C_{N} \subseteq C_{M}$,
\item[(ii)] $E(N) \cap C_{M} = C_{N}$,
\end{itemize}
which in turn imply that  $C_{M} - C_{N} \subseteq X$. Therefore, $M\backslash{X}$ is a minor of $M\backslash{\{C_{M} - C_{N}\}}$ and since  $N$ is a minor of $M\backslash{X}$ we obtain that $N$ is a minor of  $M\backslash{\{C_{M} - C_{N}\}}$. 
By
\[
M\backslash C_{M}=  M \backslash \{C_{M} - C_{N}\} \backslash C_{N} 
\]
and the fact that $N$ is a minor of  $M\backslash{\{C_{M} - C_{N}\}}$ we have that $N \backslash C_{N}$ is a minor of $M\backslash C_{M}$. 
\end{proof}
We are now ready to prove the main result of this section.

\begin{theorem}\label{th_ll3}
Let $M$ be a regular matroid such that all its cocircuits are graphic. Then, $M$ is signed-graphic if and only if $M$ has no minor isomorphic to $M^{*}(G_{17})$ or $M^{*}(G_{19})$.
\end{theorem}
\begin{proof}
The ``only if'' part is clear because of Theorem~\ref{th_ree}.
For the ``if'' part, by way of contradiction, assume that $M$ is not signed-graphic. By Theorem~\ref{th_ree}, $M$ must contain
a minor $N$ which is isomorphic to some matroid in the set
\[
\mathcal{M}=\{M^{*}(G_1),\ldots,M^{*}(G_{16}),M^{*}(G_{18}),
M^{*}(G_{20}),\ldots,M^{*}(G_{29}), R_{15}^{*}, R_{16}^{*}\}.
\]
By case analysis, verified also by the MACEK software~\cite{Hlileny:07},
it can be shown that for each  matroid  $M'\in \mathcal{M}$ there exists
a cocircuit $Y'\in \mathcal{C}(M'^{*})$ such that the matroid $M'\backslash{Y'}$ 
does contain an $M^{*}(K_{3,3})$ or an $M^{*}(K_{5})$ as a minor. Therefore, by Theorem~\ref{th_tttt},  there exists 
a cocircuit $Y_{N}\in\mathcal{C}(N^{*})$ such  that $N\backslash Y_{N}$ is not graphic. 
Therefore, by Lemma~\ref{lem_1}, there is a cocircuit $Y_{M}\in\mathcal{C}(M^{*})$ such that
$N\backslash Y_{N}$ is a minor of $M\backslash Y_{M}$.  Thus, $M\backslash Y_{M}$
is not graphic which is in contradiction with our assumption that $M$ has graphic cocircuits.
\end{proof}

%%%%%%%%%%%%%%%%%%%%%%%%%%%%%%%%%%%%%%%%%%%%%%%%%%%%%%%%%%%%%%%%%%%%%%%%%%%%%%%%%%%%%%%%%%%%%%%%%%%%%%%%%%%%%%%%
\section{A recognition algorithm}\label{sec_recal}%%%%%%%%%%%%%%%%%%%%%%%%%%%%%%%%%%%%%%%%%%%%%%%%%%%%%%%%%%%%%%%

Based on Theorem~\ref{th_ll3}, we shall provide an algorithm which given a cographic matroid $M$ with graphic cocircuits determines whether $M$ is signed-graphic or not. In order to do this, we initially consider the following problem ($P_0$):\\
%  would provide all the $3$-connected cographic matroids with graphic hyperplanes which are not signed-graphic. The problem ($P_0$) goes as follows: \\
$\mathbf{(P_0)}$: Find the members of the class $\mathcal{G}$ of $3$-connected graphs defined as follows: $G\in{\mathcal{G}}$ if the cographic matroid $M^{*}(G)$ satisfies the following two conditions: (i) $M^{*}(G)$ has a minor isomorphic to $M^{*}(G_{17})$ or $M^{*}(G_{19})$, and (ii)  for any $X\in{\mathcal{C}^{*}(M(G))}$, the matroid $M^{*}(G)\backslash{X}$ is graphic.\\
 
By duality, we obtain the following equivalent problem $(P_1)$: \\
$\mathbf{(P_1)}$: Find the members of the class $\mathcal{G}$ of $3$-connected graphs defined as follows: $G\in{\mathcal{G}}$ if the graphic matroid $M(G)$ satisfies the following two conditions: (i) $M(G)$ has a minor isomorphic to $M(G_{17})$ or $M(G_{19})$, and (ii)  for any $X\in{\mathcal{C}(M(G))}$, the matroid $M(G)/{X}$ is cographic. 
\\

Let $M'$ be a matroid isomorphic to a minor of the graphic matroid $M(G)$. Since graphic matroids are closed under minors (see Corollary~3.2.2 in~\cite{Oxley:06}), we have that there exists a graph $G'$ such that $M(G')\cong{M'}$. This implies that there exist disjoint subsets $S$ and $T$ of $E(M(G))$ such that $M(G)\backslash{S}/T\cong{M(G')}$. By well-known results regarding the minors of graphic matroids (see results 3.1.2 and 3.2.1 in~\cite{Oxley:06}), $M(G)\backslash{S}/T=M(G\backslash{S}/T)\cong{{M(G')}}$. 
By Lemma 5.3.2 in~\cite{Oxley:06}, if  $G'$ is $3$-connected and $M(G\backslash{S}/T)\cong{{M(G')}}$ then ${G'}\cong{\hat{G}}$, where $\hat{G}$ is the graph obtained from $G\backslash{S}/T$ by deleting any isolated vertices. Thus, since both $G_{17}$ and $G_{19}$ are $3$-connected, condition (i) of $(P_1)$ is equivalent to: $G$ has a $G_{17}-$ or a $G_{19}-$minor. By the dual version of Theorem~\ref{th_tttt} and due to the fact that $K_5$ and $K_{3,3}$ are $3$-connected graphs, we have that condition (ii) of $(P_1)$ is equivalent to: for any circle $X$ of $G$, the graph $G/X$ has no $K_5-$ or $K_{3,3}-$minor. Furthermore, since there is one-to-one correspondence between the circles of a graph and the circuits of the associated graphic matroid, we easily obtain the following problem $(P_2)$ which is equivalent to $(P_1)$ and, therefore, equivalent to $(P_0)$:\\
\\
$\mathbf{(P_2)}$: 
Find the members of the class $\mathcal{G}$ of $3$-connected graphs defined as follows: $G\in{\mathcal{G}}$ if $G$ satisfies the following two conditions: (i) $G$ has a $G_{17}-$ or a $G_{19}-$minor, and (ii) for any circle $X$ of $G$, the graph $G/X$ has no $K_5-$ or $K_{3,3}-$minor. 
\\

In the following Theorem~\ref{th_g171} we identify the members of $\mathcal{G}$ which have a $G_{17}-$minor and probably not a $G_{19}-$minor. Of great importance for the proof of this theorem is Theorem~\ref{th_degami} of Negami in~\cite{Negami:82}, which is a complement theorem to the well known Wheel Theorem  of Tutte (see~\cite{Tutte:61,Tutte:01}).
\begin{theorem} \label{th_degami}
Let $H$ be a graph not isomorphic to a wheel. Then a graph $G$ is $3$-connected and has $H$ as a minor if and only if $G$ can be obtained from $H$ by a sequence of the following two operations:
\begin{itemize}
\item[O1:] addition of an edge between two non-adjacent vertices, and
\item[O2:] the replacement of a vertex $v$ of degree at least $4$ by two adjacent vertices $v_1$ and $v_2$  such that each vertex formerly adjacent to $v$ becomes adjacent to exactly one of $v_1$ or $v_2$ so that in the resulting graph the degree of each  of $v_1$ and $v_2$ is greater than $2$.
\end{itemize}
\end{theorem}

Moreover, in Theorem~\ref{th_g171} we denote by $K_{3,n}$  the complete bipartite graph with $3$ and $n$ vertices at each side of the bipartition and by $K^{+1}_{3,n}$, $K^{+2}_{3,n}$ and $K^{+3}_{3,n}$ are denoted the graphs which are isomorphic to the graphs depicted in Figure~\ref{fig:k3n}, where $n\geq{5}$. Clearly, $K^{+1}_{3,n}$, $K^{+2}_{3,n}$ or $K^{+3}_{3,n}$ can be obtained from $K_{3,n}$ by a specific addition of one, two or three edges, respectively.

\begin{figure}[h] 
\begin{center}
\centering
\psfrag{(i):K1}{\footnotesize $K^{+1}_{3,n}$}
\psfrag{(ii):K2}{\footnotesize $K^{+2}_{3,n}$}
\psfrag{(iii):K3}{\footnotesize $K^{+3}_{3,n}$}
\includegraphics*[scale=0.4]{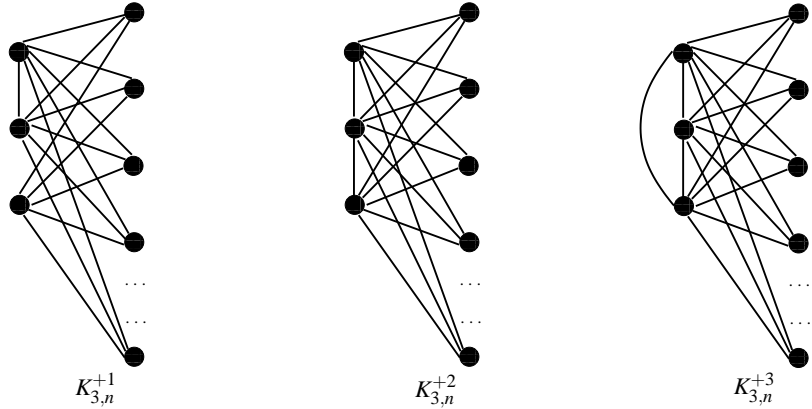}
\end{center}
\caption{The graphs $K^{+1}_{3,n}$, $K^{+2}_{3,n}$ and $K^{+3}_{3,n}$.}
\label{fig:k3n}
\end{figure}

\begin{theorem} \label{th_g171}
A graph $G$ is isomorphic to one of the graphs in  $\mathcal{L}=\{K_{3,n}$, $K^{+1}_{3,n}$, $K^{+2}_{3,n}$, $K^{+3}_{3,n}\}$, where $n\geq{5}$, if and only if $G$ satisfies the following conditions:
\begin{itemize}
\item [(i)] it is $3$-connected,
\item [(ii)]it has a $G_{17}-$minor, and
\item [(iii)] for any circle $X$ of $G$,  the graph $G/X$ has neither a $K_5$ nor a $K_{3,3}$ as a minor.
\end{itemize}
\end{theorem} 
\begin{proof}
Let $\mathcal{F}$ be the class of graphs consisting of all the graphs satisfying conditions (i), (ii) and (iii) of the theorem. We shall show that $\mathcal{F}=\mathcal{L}$. Note that for any graph in $\mathcal{L}$ we denote by $B$ the subset of its vertices having degree $3$ and by $A$ the set of the remaining vertices. Furthermore, if an $O_2$ operation is applied in a graph $H$ in $\mathcal{L}$ such that a vertex $v$ is replaced by two vertices $v_1$ and $v_2$ then in the graph $G$ so obtained we say that $A$ is the set of vertices of $G$ obtained from the set $A$ associated with $H$ by deleting $v$ and adding $v_1$ and $v_2$.  We initially prove two claims.\\
{\bf Claim 1:} Let $G_e$ be a graph obtained from a graph $G\in{\mathcal{L}}$ by applying operation $O1$ and also let  $e$ be the edge added by this operation. Then, if at least one end-vertex of $e$ is in $B$ then $G_e\notin{\mathcal{F}}$.\\
\emph{Proof of Claim 1:} 
If one of the end-vertices of $e$ is in $A$ and the other is in  $B$ then $e$ will have the same end-vertices with an existing edge of $G$ and thus, the graph $G_e$ so-obtained will not be $3$-connected. For the remaining case, up to isomorphism, $G_e$ will have as a subgraph  the graph depicted in Figure~\ref{fig:k35sub42}. Contracting the circle consisting of the dashed edges we obtain a graph having a $K_{3,3}-$minor. Thus, in both cases we have that $G_e\notin{\mathcal{F}}$. $\blacksquare$

\begin{figure}[h] 
\begin{center}
\centering
\psfrag{A}{\footnotesize $A$}
\psfrag{B}{\footnotesize $B$}
\psfrag{e}{\footnotesize $e$}
\includegraphics*[scale=0.4]{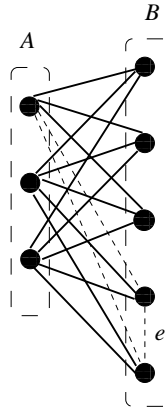}
\end{center}
\caption{Operation $O1$ applied to $K_{3,5}$.}
\label{fig:k35sub42}
\end{figure}

{\bf Claim 2:} Let $G_v$ be a graph obtained from a graph $G\in{\mathcal{L}}$ by applying operation $O2$ and also let $v_1$ and $v_2$ be the vertices of $G_v$ replacing a vertex $v$ of $G$ due to this operation. If each of $v_1$ and $v_2$ is adjacent to at least one vertex of $A$ and at least one vertex of $B$ then $G_e\notin{\mathcal{F}}$.\\
\emph{Proof of Claim 2:}
Clearly $v$ has to be a vertex of $A$ in $G$ since the vertices of $B$ have degree $3$ and thus, Operation $O2$ can not be applied. If we apply Operation $O2$ to $G$ such that each of the vertices $v_1$ and $v_2$ so-created is adjacent to at least one vertex of $A$ and at least one vertex of $B$ then it is not difficult to check that the graph $G_v$ so-obtained will have a subgraph isomorphic to one of the two graphs depicted in Figure~\ref{fig:k35sub22}. It is easy to see that the graph of  Figure~\ref{fig:k35sub22}(i) contains a $K_{3,3}-$minor.
In the remaining case, if $G_v$ has a subgraph isomorphic to the graph depicted in Figure~\ref{fig:k35sub22}(ii) then if we contract the circle consisting of the dashed edges in this graph, we obtain a graph having $K_{3,3}$ as a minor and therefore, $G_v\notin{\mathcal{F}}$. $\blacksquare$ 

\begin{figure}[h] 
\begin{center}
\centering
\psfrag{A}{\footnotesize $A$}
\psfrag{B}{\footnotesize $B$}
\psfrag{v1}{\footnotesize $v_1$}
\psfrag{v2}{\footnotesize $v_2$}
\psfrag{(i)}{\footnotesize $(i)$}
\psfrag{(ii)}{\footnotesize $(ii)$}
\includegraphics*[scale=0.4]{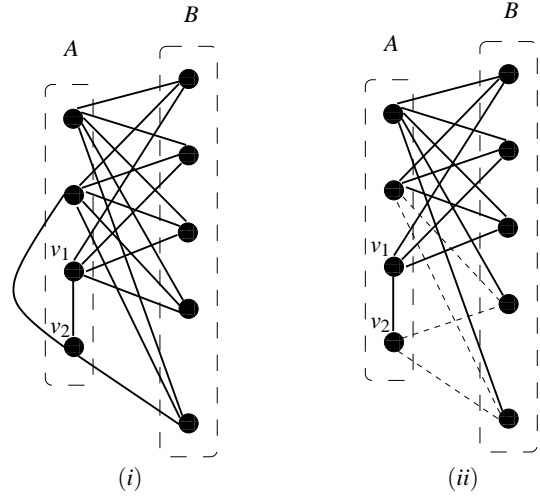}
\end{center}
\caption{Operation $O2$ applied to $K_{3,5}^{+2}$.}
\label{fig:k35sub22}
\end{figure}

Suppose now that we start from $G_{17}\cong{K_{3,5}}$ and apply Theorem~\ref{th_degami} in order to create the class of $3$-connected graphs having $K_{3,5}$ as a minor. Clearly, operations $O1$ and $O2$ have to be applied. Due to Claim~1, in order to produce a graph which may be in $\mathcal{F}$, operation $O1$ must be applied such that the new edge added joins two vertices in the vertex set $A$ of $K_{3,5}$; otherwise, a graph not in $\mathcal{F}$ is created. Because of Claim~2, in order to produce a graph which may be in $\mathcal{F}$, operation $O2$ can take place only after we have obtained a graph being $K_{3,5}^{+2}$ or $K_{3,5}^{+3}$ by a sequence of $O1$ operations; furthermore, the vertex $v$ of $K_{3,5}^{+2}$ or $K_{3,5}^{+3}$ which is replaced by operation $O2$ has to be a vertex of $A$ being adjacent to the other two vertices of $A$. Applying now operation $O2$ on $K_{3,5}^{+2}$ or $K_{3,5}^{+3}$ as described above we obtain  $K_{3,6}$ or $K_{3,6}^{+1}$, respectively, both of which belong to $\mathcal{F}$. Similarly we can apply $O1$ and $O2$ on $K_{3,6}$ or $K_{3,6}^{+1}$ in the way implied by Claims~1 and~2 in order to create $3$-connected graphs  which may belong to $\mathcal{F}$.  Continuing this process we get that all the possible $3$-connected graphs so-obtained constitute a class which is equal to $\mathcal{L}$  and includes $\mathcal{F}$.

It remains to be shown that any member of $\mathcal{L}$ is a member of $\mathcal{F}$. Clearly any graph in $\mathcal{L}$ is $3$-connected and has $G_{17}$ as a minor. We shall show that there is no circle $Y$ of a graph $H\in{\mathcal{L}}$ such that $H/X$ contains  $K_5$ or $K_{3,3}$ as a minor. We firstly show that $H$ does not contain  $K_5$ as a minor which clearly implies that $H/X$ has no $K_5-$minor. Observe that $H$ has  three vertices of degree more than $3$ and that we would like to produce by a sequence of edge and vertex deletions and edge contractions the graph $K_5$ which has five vertices of degree $4$. Clearly, the deletion of vertices or edges from a graph can not increase the degree of a vertex in the graph so-obtained; on the other hand, we can not obtain a graph from $H$ by contracting edges which will have more than three vertices with degree greater than $3$. We now show that $H/Y$ does not have $K_{3,3}$ as a minor. Observe that any cycle of $H$ contains at least two vertices belonging to the vertex set $A$ of $H$. Thus, in $H/Y$ there are at most two vertices with degree greater than $3$. Evidently, no sequence of deletions and contractions of edges and  deletions of vertices produces a graph which is isomorphic to $K_{3,3}$.  
\end{proof}

In the following Theorem~\ref{th_g191} we identify the members of $\mathcal{G}$ which have a $G_{19}-$minor and probably not a $G_{17}-$minor. As we did in Theorem~\ref{th_g171} regarding the $G_{17}$ case, we use Theorem~\ref{th_degami} in order to identify these members.

\begin{theorem} \label{th_g191}
A graph $G$ is isomorphic to one of the graphs in  $\mathcal{M}=\{K^{-}_{4,4}, K_{4,4}\}$ if and only if $G$ satisfies the following conditions:
\begin{itemize}
\item [(i)] it is $3$-connected,
\item [(ii)]it has a $G_{19}-$minor, and
\item [(iii)] for any circle $X$ of $G$,  the graph $G/X$ has neither a $K_5$ nor a $K_{3,3}$ as a minor.
\end{itemize}
\end{theorem}
\begin{proof}
Let $\mathcal{H}$ be the class of graphs containing all the graphs satisfying conditions (i), (ii) and (iii) of the theorem.  We shall show that $\mathcal{H}=\mathcal{M}$. Any graph in $\mathcal{L}$ is bipartite and thus, for a graph $G\in\mathcal{L}$ there exists a bipartition of $V(G)$ into two sets, which we denote by $V_1$ and $V_2$, such that no two adjacent vertices of $G$ belong to the same set of the bipartition.
We initially prove two claims.\\
{\bf Claim 1:} Let $G_e$ be a graph obtained from a graph $G\in{\mathcal{M}}$ by applying operation $O1$, where let $e$ be the edge added by this operation. Then, if both end-vertices of $e$ belong to either $V_1$ or $V_2$ then $G_e\notin{\mathcal{H}}$.\\
\emph{Proof of Claim 1:} 
We prove the claim only for the case in which $G$ is isomorphic to $K^{-}_{4,4}$, since the case in which $G$ is isomorphic to $K_{4,4}$ follows easily. There are two non-isomorphic graphs obtained by applying operation $O1$ such that $e$ has both of its end-vertices in $V_1$ or in $V_2$; these graphs are depicted in Figure~\ref{fig:k44ad}. In these graphs, if we contract the circle consisting of the dashed edges then we obtain a graph containing a minor isomorphic to $K_{3,3}$ and thus, the result follows. $\blacksquare$ 
\begin{figure}[h] 
\begin{center}
\centering
\includegraphics*[scale=0.4]{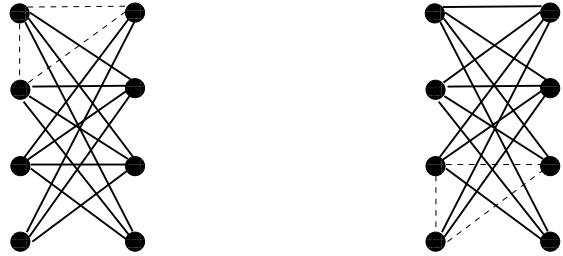}
\end{center}
\caption{Operation $O1$ applied to $K^{-}_{4,4}$.}
\label{fig:k44ad}
\end{figure}  

{\bf Claim 2:} If $G_v$ is a graph obtained from a graph $G\in{\mathcal{M}}$ by applying operation $O2$ then $G_v\notin{\mathcal{H}}$.\\
\emph{Proof of Claim 2:} 
If $G$ is $K^{-}_{4,4}$ then any possible application of operation $O2$ to this graph would produce a graph being isomorphic to that depicted in (i) of Figure~\ref{fig:k44spl}. Similarly, if $G$ is $K_{4,4}$ then any possible application of operation $O2$ to this graph would produce a graph being isomorphic to that depicted in (ii) of Figure~\ref{fig:k44spl}. In these graphs, if we contract the cycle consisting of the dashed edges then, in each case, we obtain a graph which has $K_{3,3}$ as a minor and thus, the result follows. $\blacksquare$ 
\begin{figure}[h] 
\begin{center}
\centering
\psfrag{(i)}{\footnotesize $(i)$}
\psfrag{(ii)}{\footnotesize $(ii)$}
\includegraphics*[scale=0.4]{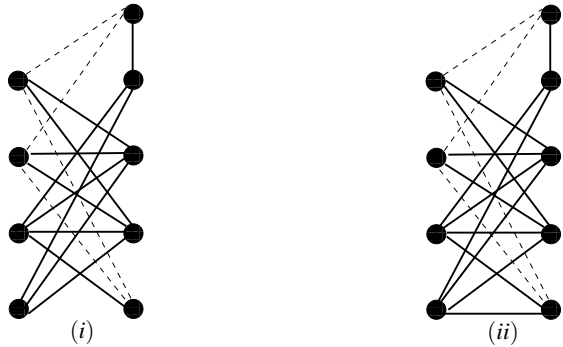}
\end{center}
\caption{Operation $02$ applied to $K^{-}_{4,4}$ and  $K_{4,4}$.}
\label{fig:k44spl}
\end{figure}

In order to produce all the $3$-connected graphs having $G_{19}\cong{K^{-}_{4,4}}$ as a minor Theorem~\ref{th_degami} can be applied. However, due to Claim~2, if we apply operation $O2$ on $K^{-}_{4,4}$ or $K_{4,4}$ then the graph so-obtained is not in $\mathcal{H}$. Moreover, due to Claim~1, Operation $O_1$ may only add an edge joining vertices between different vertices of the vertex bipartition of $K^{-}_{4,4}$ in order to obtain a graph in $\mathcal{H}$. Thus, it is clear that the class of graphs satisfying conditions (i), (ii) and (iii) of the Theorem is a subclass of $\mathcal{M}$. We shall now complete the proof by showing that $K^{-}_{4,4}$ and $K_{4,4}$ are in $\mathcal{H}$. Clearly both are $3$-connected and contain $G_{19}\cong{K^{-}_{4,4}}$ as a minor. Furthermore, any cycle in both graphs consists of at least four vertices and therefore, the contraction of the associated circle gives rise to a graph with at most five vertices which clearly can not have $K_{3,3}$ as a minor. Similarly, it can be easily checked that none of $K^{-}_{4,4}$ and $K_{4,4}$ has $K_{5}$ as a minor. 
\end{proof}

We are ready now to prove the main theorem of this section.
\begin{theorem} \label{th_manama}
A cographic matroid $M$ satisfies the following conditions:
\begin{itemize}
\item[(i)] it is $3$-connected,
\item[(ii)] it has a minor isomorphic to $M^{*}(G_{17})$ or $M^{*}(G_{19})$, and
\item[(iii)] for any $Y\in{\mathcal{C}^{*}(M)}$, $M\backslash{Y}$ is graphic 
\end{itemize}
if and only if $M\cong{M^{*}(G)}$, where $G\in{\{K_{3,n},K^{+1}_{3,n},K^{+2}_{3,n},K^{+3}_{3,n},K^{-}_{4,4}, K_{4,4} \}}$ with $n\geq{5}$.
\end{theorem}
\begin{proof}
By  Theorems ~\ref{th_g171} and~\ref{th_g191}, we conclude that the solution to the problem $P_2$ are the graphs in $\mathcal{G}={\{K_{3,n},K^{+1}_{3,n},K^{+2}_{3,n},K^{+3}_{3,n},K^{-}_{4,4}, K_{4,4} \}}$ where  $n\geq{5}$. Since the problem $P_0$ is equivalent to $P_2$ we have that the ``if'' part follows.

For the ``only if'' part, clearly the cographic matroids of the graphs in $\mathcal{G}$ defined in problem $P_0$ are those satisfying the conditions of the theorem. We have shown that $P_0$ is equivalent to $P_2$. Thus, by Theorems ~\ref{th_g171} and~\ref{th_g191}, we have that if $M$ satisfies conditions (i), (ii) and (iii) then $M$ is isomorphic to the cographic matroid associated with a graph in $\mathcal{L}\cup{\mathcal{M}}=\{K_{3,n},K^{+1}_{3,n},K^{+2}_{3,n},K^{+3}_{3,n},K^{-}_{4,4}, K_{4,4} \}$.
\end{proof}

We are now ready to present an algorithm  which given a cographic matroid $M$ with graphic cocircuits determines whether $M$ is signed-graphic or not.

\noindent {\bf \sc{Recognition Algorithm}}
\newline 

{\bf Input:} A cographic matroid $M$ with graphic cocircuits.
 
{\bf Output:} The matroid $M$ is identified as signed-graphic or not. 
\newline

{\bf Step 1.} Decompose $M$  into 3-connected cographic minors $M_1,\ldots,M_l$ ($l\geq{1}$)of $M$  via $1-$ and $2-$sums using the decomposition algorithm provided  in~\cite{Truemper:85,Truemper:98}.

{\bf Step 2.} For each $M_i$ $(i=1,\ldots, l)$, construct the unique up to isomorphism graph $H_i$ such that $M_i=M(H_i)$. This can be done by applying the algorithm appearing in~\cite{Tutte:1960,BixCunn:1980}.

{\bf Step 3.} Test if there exists an $H_i$ being isomorphic to one of the graphs in $\mathcal{G}=\{K_{3,n},K^{+1}_{3,n},K^{+2}_{3,n},K^{+3}_{3,n},K^{-}_{4,4}, K_{4,4} \}$. If yes, then $M$ is not signed-graphic; otherwise, $M$ is signed-graphic.
\newline

Let $A$ be an $m\times{n}$ binary matrix such that $M\cong{M[A]}$ and let $w(A)$ be the number of nonzeros of $A$. Then, there exists an $O((m+n)\cdot{w(A)})$ time algorithm for step 1 (see \cite{Truemper:90}) and an $O(m\cdot w(A))$ time algoritm for step 2 (see \cite{BixCunn:1980}). Checking whether a graph is isomorphic to some graph in $\mathcal{G}$ is easy (i.e. it can be carried out in polynomial time) due to the special structure of the graphs in $\mathcal{G}$; specifically, it is trivial to check if a graph $H_i$ given by step 2 is isomorphic with $K^{-}_{4,4}$ or $K_{4,4}$ while $H_i$ must have $n$ mutually non-adjacent vertices of degree $3$ and a particular adjacency relation between the remaining $3$ vertices in order to be isomorphic to a graph in $\{K_{3,n},K^{+1}_{3,n},K^{+2}_{3,n},K^{+3}_{3,n}\}$.    Regarding the storage of graphs and matrices, the simple data structures used for graphs and matrices in~\cite{Truemper:90} are employed. The proof of correctness of this algorithm goes as follows. 
Since $M^{*}(G_{17})$ and $M^{*}(G_{19})$ are $3$-connected matroids we have that if such a matroid was a minor of $M$ then it must also be a minor of some matroid $M_j$ in $\{M_1,\ldots,M_l\}$ (see~\cite{Oxley:06,Truemper:98}). Therefore, $M$ has an $M^{*}(G_{17})-$ or an $M^{*}(G_{19})-$minor if and only if some $M_j$ in $\{M_1,\ldots,M_l\}$ has such a minor. Since, by Lemma~\ref{lem_1}, ``having graphic cocircuits'' is a minor-closed property we have that each of $M_1,\ldots,M_l$ has graphic cocircuits. Thus, by Theorem~\ref{th_manama}, $M_j$ has a minor isomorphic to $M^{*}(G_{17})$ or $M^{*}(G_{19})$ if and only if $M_j$ is isomorphic to some cographic matroid associated with a graph in $\{K_{3,n},K^{+1}_{3,n},K^{+2}_{3,n},K^{+3}_{3,n},K^{-}_{4,4}, K_{4,4} \}$ (where  $n\geq{5}$). 

\bibliographystyle{eptcs}

\end{document}